\documentclass[12pt]{amsart}
\usepackage{amsfonts}
\usepackage{graphicx}
\usepackage{amsmath}
\usepackage{amssymb}
\input{amssym.def}
\newtheorem{theorem}{Theorem}
\newtheorem{proposition}{Proposition}

\newtheorem{lemma}{Lemma}
\newtheorem{remark}{Remark}

\theoremstyle{remark}

\newcommand{\re}{\text{\rm Re }}

\newcommand{\Sb}{\text{\bf S}}
\newcommand{\Diff}{\text{\rm Diff }}

\newcommand{\Vect}{\text{\rm Vect }}
\newcommand{\const}{\text{\rm const}}
\newcommand{\s}{\vspace{0.3cm}}

\DeclareMathOperator{\p}{\partial}

\input epsf
\begin{document}
\title[Virasoro algebra in L\"owner-Kufarev dynamics]{Virasoro Algebra in L\"owner-Kufarev contour dynamics}
\author[I.~Markina and A.~Vasil'ev]{Irina Markina and Alexander Vasil'ev}

\thanks{The authors were  supported by the grant of the Norwegian Research Council \#177355/V30, by the NordForsk network `Analysis and Applications' grant \#080151, and by the European Science Foundation Research Networking Programme HCAA}
\address{Department of Mathematics, University of Bergen, Johannes Brunsgate 12, Bergen 5008, Norway}
\email{irina.markina@math.uib.no}
\email{alexander.vasiliev@math.uib.no}

\begin{abstract}
Contour dynamics is a classical subject both in physics and in complex analysis.
We show that the dynamics provided by the L\"owner-Kufarev ODE and PDE possesses a rigid algebraic structure given by the Virasoro algebra.
Namely, the `positive' Virasoro generators span the holomorphic part of the complexified vector bundle over the space of univalent functions, smooth on the boundary. In the covariant formulation they are conserved by the L\"owner-Kufarev evolution.
The `negative' Virasoro generators span the antiholomorphic part.  They contain a conserved term and we give an iterative method to obtain them based on the Poisson structure of the L\"owner-Kufarev evolution. The L\"owner-Kufarev PDE provides a distribution  of the tangent bundle of non-normalized univalent functions, which forms the tangent bundle of normalized ones. It also gives an explicit
correspondence between the latter bundle and the holomorphic eigen space of the complexified Lie algebra of vector fields on the unit circle.
Finally, we give Hamiltonian and Lagrangian formulations of the motion within the coefficient body in the field of an elliptic operator constructed by means of  Virasoro generators. We also discuss relations between CFT and SLE.
\end{abstract}
\maketitle

\section{Introduction}

The challenge of structural understanding of non-equilibrium interface dynamics has become increasingly important in mathematics and physics. Dynamical interfacial properties, such as fluctuations, nucleation and aggregation, mass and charge transport, are often very complex. There exists no single theory or model that can predict all such properties. Many physical processes, as well as complex dynamical systems, iterations and construction of Lie semigroups with respect to the composition operation, lead to the study of growing systems of plane domains. Recently, it has become clear that one-parameter expanding evolution families of simply
connected domains in the complex plane in some special models has been governed by infinite systems
of evolution parameters, conservation laws. This phenomenon reveals a bridge between a non-linear evolution of complex shapes emerged in physical problems, dissipative in most of the cases, and exactly solvable models. A sample problem is the Laplacian growth, in which the harmonic (Richardson's) moments are conserved under the evolution, see e.g., \cite{Mineev, Vas}. The infinite number of evolution parameters reflects the infinite number of degrees of freedom of the system, and clearly suggests to apply field theory methods as a natural tool of study. The Virasoro algebra provides a structural background in most of field theories, and it is not surprising that it appears in soliton-like problems, e.g., KdV or Toda hierarchies, see \cite{Faddeev, Gervais}.

Another group of models, in which the evolution is governed by an infinite number of parameters, can be observed in controllable dynamical systems, where the infinite number of degrees of freedom follows from the infinite number of driving terms. Surprisingly,
the same structural background appears again for this group. We develop this viewpoint in the present paper.

One of the general approaches to the growing contour evolution was provided by L\"owner and Kufarev  \cite{Loewner, Pommerenke2}. The contour evolution is described by a time-dependent conformal parametric map from a canonical domain, the unit disk in most of the cases, onto the domain
bounded by the contour for each fixed instant. In fact, these one-parameter conformal maps satisfy the L\"owner-Kufarev partial differential equation. A characteristic equation to this PDE represents an infinite dimensional controllable system for which
the infinite number of conservation laws is given by the Virasoro generators in their covariant form.  

Recently, Friedrich and
Werner \cite{FriedrichWerner}, and independently Bauer and Bernard \cite{BB}, found relations between SLE (stochastic or Schramm-L\"owner evolution) and the highest weight representation of the Virasoro algebra. 

All  above results encouraged us to conclude that the Virasoro algebra is a common structural basis for these and possibly other types of contour dynamics and we present the development in this direction here. For the first time, a construction, which appeared in the field theory 
plays the algebraic structural  background for the contour evolution in classical complex analysis.

The structure of the paper is as follows. Sections 2 and 3 contain the necessary background on the Virasoro algebra and the L\"owner-Kufarev equations. The main results are contained in Sections 4 and 5.  In Section 4 we construct the Poisson structure on the cotangent bundle of the space of univalent functions smooth on the boundary and the Hamiltonian system generated by the L\"owner-Kufarev equation in ordinary derivatives.
We establish that the holomorphic Virasoro generators in the covariant formulation are conserved under the L\"owner-Kufarev evolution (Theorem 2).
The antiholomorphic generators are proved to contain a conserved term and we give an iterative method to obtain them based on the Poisson structure of the L\"owner-Kufarev evolution. The L\"owner-Kufarev PDE is shown to provide a distribution  of the tangent bundle of non-normalized univalent functions, which forms the tangent bundle of normalized ones. It also gives an explicit
correspondence between the latter bundle and the holomorphic eigen space of the complexified Lie algebra of vector fields on the unit circle.
In Section 5, we give Hamiltonian and Lagrangian formulations of the motion within the coefficient body in the field of an elliptic operator constructed by means of  Virasoro generators.   The solutions with constant velocity coordinates are found. We prove that the norm of the driving function in the L\"owner-Kufarev theory gives the minimal energy of the motion. The short Section 6 we add for completeness. We briefly review the connections between conformal field theory and the Schramm-L\"owner evolution following \cite{BB, FriedrichWerner}.

\s
\noindent
{\bf Acknowledgements.}
We are thankful to H\'el\`ene Airault, Ludwig Faddeev, Paul Malliavin, and Yurii Neretin for many helpful discussions concerning the Virasoro algebra and its representations.

\section{Virasoro Algebra}

The Virasoro algebra $Vir$ plays a prominent role in modern mathematical physics, both in field theories and  solvable models. It appears in physics literature as an algebra obeyed by the stress-energy tensor and associated with the conformal group, the Virasoro-Bott group, of the worldsheet in two dimensions, see e.g., \cite{Polchinski}. It  is a unique central extension of the Lie algebra for the  Lie-Fr\'echet group $\Diff S^1$ of sense-preserving diffeomorphisms of the unit circle $S^1$, and it is an infinite-dimensional real vector space. The extension is characterized by a real parameter $c$, so the Virasoro algebra refers to a class of isomorphic Lie algebras corresponding to different values of $c$. At the same time the Virasoro algebra is intrinsically related to the KdV canonical structure where the Virasoro brackets become the Magri brackets for the Miura transformations of  elements of the phase space of the KdV hierarchy (see, e.g., \cite{Faddeev, Gervais}). 
 
The complex hull $\mathbb CVir$ of the Virasoro algebra can be realized as a central extension by $\mathbb C$ of the Witt algebra, a complex Lie algebra of derivations (or Leibnitz rule) of  the algebra $\mathbb C[z,z^{-1}]$ of complex Laurent polynomials. The Witt algebra  is spanned by the generators $L_n=z^{n+1}\frac{\partial}{\partial z}$ on $\mathbb C\setminus \{0\}$. The operators $L_n$ plus a central element $c$ are called the Virasoro generators. Under any irreducible representation of $\mathbb CVir$, the quantity $c$
is realized as a complex scalar and is called the central charge. The generators satisfy the commutation relations given by  
$$\{L_m,L_n\}_{Vir}=
  (n-m)L_{m+n}+\frac{c}{12}n(n^2-1)\delta_{n,-m},\quad \{L_n,c\}_{Vir}=0,\quad n,m\in \mathbb Z,
$$
where $c\in \mathbb C$ is the central charge. Considering the Virasoro algebra as an operator algebra, the generators $L_n$ become the coefficients in a formal Laurent series for the analytic component of the stress-energy tensor in 2-D field theory.  The attribution `Virasoro algebra' is due to a Virasoro's seminal  paper \cite{Virasoro}. 

Mathematically, the Virasoro algebra appeared for the first time  as a central  extension by the {\it Gelfand-Fuchs cocycle}~\cite{GelfandFuchs} of the Lie algebra $\Vect S^1$ of smooth vector fields $\phi\frac{d}{d\theta}$ on the unit circle   $S^1$ (see \cite{GelfandFuchs}), where the Lie bracket is defined to be the commutator of vector fields
\begin{equation}\label{Lie1}
[\phi_1,\phi_2]={\phi}_1{\phi}'_2-{\phi}_2{\phi}'_1.
\end{equation}
  Each element of the  Lie-Fr\'echet group  $\Diff S^1$ is represented as
$z=e^{i\alpha(\theta)}$ with a monotone increasing $C^{\infty}$
real-valued function $\alpha(\theta)$, such that
$\alpha(\theta+2\pi)=\alpha(\theta)+2\pi$.
The Lie algebra for this group is identified with $\Vect S^1$. The relation of this Lie algebra to $\Diff S^1$ is subtile because  the exponential map is not even locally a homeomorphism.

\subsection{Canonical identification}
The entire necessary background of unitary representations of $\Diff S^1$ is found in the study of Kirillov's
homogeneous K\"ahlerian manifold $\Diff S^1/S^1$.
We  deal with the analytic representation of 
$\Diff S^1/S^1$. Let $\Sb$ stand for the whole class of univalent functions $f$ in the unit disk $U$ normalized by $f(z)=z(1+\sum_{n=1}^{\infty}c_nz^n)$ about the origin and $C^{\infty}$-smooth on the boundary $S^1$ of~$U$.
Given a map $f\in \Sb$ we construct the adjoint univalent
meromorphic map 
\[
g(z)=d_1z+d_0+\frac{d_{-1}}{z}+\dots,
\] 
defined in the exterior $U^*=\{z:\,|z|>1\}$ of $U$, and such that  $\hat{\mathbb C}\setminus\overline{f(U)}=g(U^*)$. Both functions are extendable onto $S^1$. This conformal welding gives the identification of the homogeneous manifold $\Diff S^1/S^1$ with the space $\Sb$:  $\Sb\ni f\leftrightarrow f^{-1}\circ g|_{S^1}\in\Diff S^1/S^1$, or with the  smooth contours $\Gamma=f(S^1)$ that enclose univalent domains $\Omega$ of conformal radius 1 with respect to the origin and such that $\infty\not\in \Omega$, $0\in\Omega$, see \cite{Airault}, \cite{KY1}.  So one can construct complexification
of $\Vect S^1$ and further projection of the holomorphic part to the set $\mathcal M\subset \mathbb C^{\mathbb N}$, which is the projective limit of the coefficient bodies $\mathcal M=\lim_{n\leftarrow \infty}\mathcal M_n$, where
\begin{equation}
\mathcal M_n=\{(c_1,\dots,c_n):\,\,f\in \Sb\}.\label{Mn}
\end{equation}
The holomorphic Virasoro generators can then be realized by the first order differential operators
\[
L_j=\partial_j+\sum\limits_{k=1}^{\infty}(k+1)c_{k}\partial_{j+k},\quad j\in \mathbb N,
\]
in terms of the affine coordinates of $\mathcal M$, acting over the set of holomorphic functions, where $\partial_{k}=\partial/\partial{c_k}$. We explain the details in the next subsection.

\subsection{Complexification}
Let us introduce local coordinates on the manifold \linebreak $\mathcal M=\Diff S^1/S^1$ in the concordance with the local coordinates on the space  $\Sb$ of univalent functions smooth on the boundary. Observe that $\mathcal M$ is a real infinite-dimensional manifold, whereas $\Sb$
is a complex manifold. We are aimed at a complexification of $T\mathcal M$ which admits a holomorphic projection to $T\Sb$, where 
$\Vect_0 S^1=\Vect\, S^1/\const$ is a module over the ring of smooth functions, which is associated with the tangent bundle $T\mathcal M$. 

 Given a real vector space $V$ the complexification $V_{\mathbb C}$ is defined as the tensor product with the complex numbers $V\otimes_{\mathbb R}\mathbb C$. Elements of $V_{\mathbb C}$ are of the form $v\otimes z$.
In addition, the  vector space $V_{\mathbb C}$ is a complex vector space that follows by defining multiplication by complex numbers,
$\alpha(v\otimes z)=v\otimes \alpha z$ for complex $\alpha$ and $z$ and $v\in V$. The space $V$ is naturally embedded into $V\otimes \mathbb C$ by identifying $V$ with $V\otimes 1$.
 Conjugation
is defined by introducing a canonical conjugation map on $V_{\mathbb C}$ as $\overline{v\otimes z}=v\otimes \bar z$.

An almost complex structure $J$ on $V$ can be extended by linearity to the complex structure $J$ on $V_{\mathbb C}$ by $J(v\otimes z)=J(v)\otimes z$. Observe that 
\[
\overline{J(v\otimes z)}=J(\overline{v\otimes z}).
\]

Eigenvalues of extended $J$ are $\pm i$, and there are two eigenspaces $V^{(1,0)}$ and $V^{(0,1)}$ corresponding to them given by projecting $\frac{1}{2}(1\mp iJ)v$. $V_{\mathbb C}$ is decomposed into the direct sum $V_{\mathbb C}=V^{(1,0)}\oplus V^{(0,1)}$, where $V^{(1,0)}=\{v\otimes 1- J(v)\otimes i\big| v\in V \}$ and $V^{(0,1)}=\{v\otimes 1+ J(v)\otimes i\big| v\in V \}$ are the eigen spaces corresponding to $\pm i$.

An almost complex structure on $\Vect_0 S^1$ may be defined as follows (see \cite{Airault}).
We  identify $\Vect_0 S^1$  with the functions with vanishing mean value over~$S^1$. It gives
\[
\phi(\theta)=\sum\limits_{n=1}^{\infty}a_n\cos\,n\theta+b_n\sin\,n\theta.
\]
Let us define an almost complex structure by the operator
\begin{equation}\label{CompStruct}
J(\phi)(\theta)=\sum\limits_{n=1}^{\infty}-a_n\sin\,n\theta+b_n\cos\,n\theta.
\end{equation}
 On $\Vect_0S^1\otimes \mathbb C$, the operator $J$ diagonalizes and we have the identification
\[
\Vect_0S^1\ni \phi\leftrightarrow v:=\frac{1}{2}(\phi-iJ(\phi))=\sum\limits_{n=1}^{\infty}(a_n-ib_n)e^{in\theta}\in (\Vect_0S^1\otimes \mathbb C)^{(1,0)},
\]
and the latter extends into the unit disk as a holomorphic function.

The Kirillov infinitesimal action \cite{Kir2} of $\Vect_0 S^1$ on $\Sb$ is given by
a variational formula due to Schaeffer and Spencer \cite[page 32]{Schaeffer} which lifts the actions from the Lie algebra $\Vect_0
S^1$ onto $\Sb$. Let $f\in\Sb$ and let
$\phi(e^{i\theta}):=\phi(\theta)\in \Vect_0 S^1$ be a $C^{\infty}$ real-valued function in
$\theta\in(0,2\pi]$. The infinitesimal action
$\theta \mapsto \theta+\varepsilon \phi(e^{i\theta})$ yields a variation of the univalent function $f^*(z)= f+\varepsilon\,\delta_{v}f(z)+o(\epsilon)$, where
\begin{equation}
\delta_{v}f(z)=\frac{f^2(z)}{2\pi
}\int\limits_{S^1}\left(\frac{wf'(w)}{f(w)}\right)^2\frac{v(w)dw}{w(f(w)-f(z))} ,\label{var}
\end{equation}
and $\phi\leftrightarrow v$ by the above identification.
 Kirillov and Yuriev \cite{KY1}, \cite{KY2} (see also \cite{Airault})  established
that the variations $\delta_{\phi}f(\zeta)$ are closed with respect to
the commutator~(\ref{Lie1}), and the induced Lie algebra is the same as $\Vect_0
S^1$.  The Schaeffer-Spencer operator is linear.

Treating $T\mathcal M$ as a real vector space, the operator $\delta_{\phi}$ transfers the complex structure $J$ from $\Vect_0 S^1$ to $T\mathcal M$ by $J(\delta_{\phi}):=\delta_{J(\phi)}$. By abuse of notation,  we denote the new complex structure on $T\mathcal M$  by the same character $J$. Then it splits the complexification $T\mathcal M_{\mathbb C}$ into two eigenspaces $T\mathcal M_{\mathbb C}=T\mathcal M^{(1,0)}\oplus T\mathcal M^{(0,1)}$. Therefore,  $\delta_{v}=\delta_{\phi-iJ(\phi)}:=\delta_{\phi}-iJ(\delta_{\phi})\in T\mathcal M^{(1,0)}$. 
Observe that $2z\partial_z=-i\partial_{\theta}$ on the unit circle $z=e^{i\theta}$, and $L_k=z^{k+1}d/d z=-\frac{1}{2}ie^{ik\theta}d/d \theta$ on $S^1$. Let us take the basis
of $\Vect_0 S^1\otimes \mathbb C$ in the form $\nu_k=-ie^{ik\theta}$ in order to keep the index of vector fields the same as for $L_k$. Then, the commutator satisfies the Witt relation $\{\nu_m, \nu_n\}=(n-m)\nu_{n+m}$. Taking  elements $\nu_k=-iw^k$, $|w|=1$ in the integrand of~(\ref{var}) we calculate the residue in (\ref{var}) and obtain so called Kirillov operators
\[
L_j[f](z)=\delta_{\nu_j}f(z)=z^{j+1}f'(z), \quad j=1,2,\dots,
\]
so that these $L_j$ are the  holomorphic coordinates on $T\mathcal M^{(1,0)}$.
In terms of the affine coordinates in $\mathcal M$ we get the Kirillov operators as
\[
L_j=\partial_j+\sum\limits_{k=1}^{\infty}(k+1)c_{k}\partial_{j+k},
\]
where $\partial_k=\partial/\partial c_k$. They satisfy the Witt commutation relation $$\{L_m,L_n\}=(n-m)L_{n+m}.$$
For $k=0$ we obtain the operator $L_0$, which corresponds to the constant vectors from $\Vect \,S^1$, $L_{0}[f](z)=zf'(z)-f(z)$.
The elements of  the Fourier basis $-ie^{-i\theta k}$ with negative indices (corresponding to $T\mathcal M^{(0,1)}$) are extended into $U$ by
$-iz^{-k}$. Substituting them in (\ref{var}) we get very complex formulas for $L_{-k}$, which functionally depend on $L_k$ (see \cite{Airault}, \cite{Kir2}), and  which are dual to $L_k$ with respect to  the action of $J$. The first two operators are calculated as
\begin{eqnarray*}
L_{-1}[f](z)&=&f'(z)-2c_1 f(z)-1,\\
L_{-2}[f](z)&=&\frac{f'(z)}{z}-\frac{1}{f(z)}-3c_1+(c_1^2-4c_2)f(z),
\end{eqnarray*}
see \cite{KY2}.

This procedure gives a nice links between representations of the Virasoro algebra and the theory of univalent functions.
The L\"owner-Kufarev equations proved to be a powerful tool to work with univalent functions (the famous Bieberbach conjecture was proved \cite{Branges} using L\"owner method). In the following section we show how L\"owner-Kufarev equations can be used in a representation of the Virasoro algebra. In particular, we identify $T\mathcal M^{(1,0)}$ with  $T\mathcal M$, equipped with its natural complex structure given by coefficients of univalent functions, by means the L\"owner-Kufarev PDE.

\section{L\"owner-Kufarev Equations}

A time-parameter family $\Omega(t)$ of simply connected hyperbolic univalent domains forms a {\it L\"owner subordination chain}  in the complex plane $\mathbb
C$, for $0\leq t< \tau$ (where $\tau$ may be $\infty$), if
$\Omega(t)\varsubsetneq \Omega(s)$, whenever $t<s$.
We
suppose that the origin is an interior point of the Carath\'eodory kernel of
$\{\Omega(t)\}_{t=0}^{\tau}$.  

A L\"owner subordination chain $\Omega(t)$ is described by a time-dependent family of conformal maps $z=f(\zeta,t)$
from the unit disk $U=\{\zeta:\,|\zeta|<1\}$ onto $\Omega(t)$, normalized by $f(\zeta,t)=a_1(t)\zeta+a_2(t)\zeta^2+\dots$,
$a_1(t)>0$, $\dot{a}_1(t)>0$. After L\"owner's 1923 seminal  paper \cite{Loewner} a fundamental contribution to
the theory of L\"owner chains was made by Pommerenke \cite{Pommerenke1, Pommerenke2} who
described governing evolution equations in partial and ordinary derivatives, known now as
the L\"owner-Kufarev equations due to Kufarev's work \cite{Kufarev}.

One can normalize the growth of 
evolution of a subordination chain by the conformal radius of
$\Omega(t)$ with respect to the origin by  $a_1(t)=e^t$.

 L\"owner \cite{Loewner}  studied a
time-parameter semigroup of conformal one-slit maps of the unit disk $U$ arriving
then at an evolution equation called after him. His main
achievement was an infinitesimal description of the semi-flow of
such maps by the Schwarz kernel that led him to the L\"owner
equation. This crucial result was then generalized in several
ways (see \cite{Pommerenke2} and the references therein).

We say that the  function $p$ is from the Carath\'eodory class if it is analytic in $U$, normalized as $p(\zeta)=1+p_1\zeta+p_2\zeta^2+\dots,\quad
\zeta\in U,$ and such that $\re p(\zeta)>0$ in~$U$.
Pommerenke \cite{Pommerenke1, Pommerenke2} proved that given a subordination
chain of domains $\Omega(t)$ defined for $t\in [0,\tau)$, there exists
a function $p(\zeta,t)$, measurable in $t\in [0,\tau)$ for any fixed $z\in U$, and from the Carath\'eodory class for almost all   $t\in [0,\tau)$, such that 
the conformal mapping $f:U\to \Omega(t)$ solves the equation
\begin{equation}
\frac{\partial f(\zeta,t)}{\partial t}=\zeta\frac{\partial
f(\zeta,t)}{\partial \zeta}p(\zeta,t),\label{LK}
\end{equation}
for $\zeta\in U$ and for almost all $t\in [0,\tau)$.   The
equation (\ref{LK}) is called the L\"owner-Kufarev equation due to
two seminal papers: by L\"owner \cite{Loewner} who considered the case when
\begin{equation}
p(\zeta,t)=\frac{e^{iu(t)}+\zeta}{e^{iu(t)}-\zeta},\label{yadro}
\end{equation}
where $u(t)$ is a continuous function regarding to $t\in [0,\tau)$,
 and by Kufarev \cite{Kufarev} 
who proved differentiability of $f$ in $t$ for all $\zeta$ from the kernel of $\{\Omega(t)\}$ in the case of general $p$ in the Carath\'eodory class.

 Let us consider a reverse process. We are given an
initial domain $\Omega(0)\equiv \Omega_0$ (and therefore, the
initial mapping $f(\zeta,0)\equiv f_0(\zeta)$), and a function
$p(\zeta,t)$ of positive real part normalized by $p(\zeta,t)=1+p_1\zeta+\dots$. Let us solve the equation (\ref{LK})
and ask ourselves, whether the solution $f(\zeta,t)$ defines a subordination
chain of simply connected univalent domains $f(U,t)$. The initial condition
$f(\zeta,0)=f_0(\zeta)$ is not given on the characteristics of the
partial differential equation (\ref{LK}), hence the solution exists
and is unique but not necessarily univalent. Assuming $s$ as a parameter along the characteristics
we have $$ \frac{dt}{ds}=1,\quad \frac{d\zeta}{ds}=-\zeta
p(\zeta,t), \quad \frac{df}{ds}=0,$$ with the initial conditions
$t(0)=0$, $\zeta(0)=z$, $f(\zeta,0)=f_0(\zeta)$, where $z$ is in
$U$.  Obviously, $t=s$. Observe that the domain of $\zeta$ is the entire unit disk. However, the solutions to
the second equation of the characteristic system range within the unit disk but do not fill it. 
Therefore, introducing another letter $w$ (in order to distinguish the function $w(z,t)$  from the variable $\zeta$) we arrive at the Cauchy problem for the  L\"owner-Kufarev
equation in ordinary derivatives
\begin{equation}
\frac{dw}{dt}=-wp(w,t),\label{LKord}
\end{equation}
 for a function $\zeta=w(z,t)$
with the initial condition $w(z,0)=z$. The equation (\ref{LKord}) is a non-trivial  characteristic
equation for (\ref{LK}). Unfortunately, this approach requires the
extension of $f_0(w^{-1}(\zeta,t))$ into the whole $U$ ($w^{-1}$ means the inverse function) because the solution to
(\ref{LK}) is the function $f(\zeta,t)$  given as
$f_0(w^{-1}(\zeta,t))$, where $\zeta=w(z,s)$ is a solution of the
initial value problem for the characteristic equation (\ref{LKord})
that maps $U$ into $U$. Therefore, the solution of the initial
value problem for the equation (\ref{LK}) may be non-univalent.

 Solutions to the
equation (\ref{LKord}) are holomorphic univalent functions
$w(z,t)=e^{-t}z+a_2(t)z^2+\dots$ in the unit disk that map $U$  into
itself. Every function $f$ from the class $\Sb$ can be
represented by the limit
\begin{equation}
f(z)=\lim\limits_{t\to\infty}e^t w(z,t),\label{limit}
\end{equation}
where $w(z,t)$ is a solution to \eqref{LKord}  with some  function $p(z,t)$ of positive real part for almost
all $t\geq 0$ (see \cite[pages 159--163]{Pommerenke2}). Each function
$p(z,t)$ generates a unique function from the class $\Sb$. The
reciprocal statement is not true. In general, a function $f\in \Sb$
can be obtained using different functions $p(\cdot,t)$. 

Now we are ready to formulate the condition of  univalence of the solution to the equation (\ref{LK}), which
can be obtained by combination of known results of \cite{Pommerenke2}.

\begin{theorem}\label{ThProkhVas}{\rm \cite{Pommerenke2, ProkhVas}} Given  a function
$p(\zeta,t)$ of positive real part normalized by $p(\zeta,t)=1+p_1\zeta+\dots$, the solution to   the equation (\ref{LK})
is unique, analytic and univalent with respect to $\zeta$ for almost all $t\geq 0$, if and only if, the initial condition
$f_0(\zeta)$ is taken in the form \eqref{limit}, where the function $w(\zeta,t)$ is the solution to the equation \eqref{LKord}
with the same driving function $p$.
\end{theorem}

Recently, we started to look at L\"owner-Kufarev equations from the point of view of  motion in the space of univalent functions where Hamiltonian and Lagrangian formalisms play a central role (see, \cite{Vasiliev}). Some connections with the Virasoro algebra were also observed in \cite{MPV, Vasiliev}. The present paper generalizes these attempts and gives their closed form.  The main conclusion is that the L\"owner-Kufarev equations are naturally linked to the holomorphic part of the Virasoro algebra.
Taking holomorphic Virasoro generators $L_n$ as a basis of the tangent space to the  coefficient body for univalent functions at a fixed point, we see that the driving function in the L\"owner-Kufarev theory generates generalized moments for motions within
the space of univalent functions. Its norm represents the energy of this motion. The holomorphic Virasoro generators in their co-tangent form will become  conserved quantities of the L\"owner-Kufarev ODE. The L\"owner-Kufarev PDE becomes a transition formula from the affine basis to Kirillov's basis of the holomorphic part of the complexified tangent space to $\mathcal M$ at any point.  Finally, we propose to study an alternate L\"owner-Kufarev evolution
instead of subordination.

\section{Witt algebra and the classical L\"owner-Kufarev equations}

In the following subsections we reveal the structural role of the Witt algebra as a background of the classical L\"owner-Kufarev contour evolution. As we see further,
the conformal anomaly and the Virasoro algebra appear as a quantum or stochastic effect in SLE.

\subsection{L\"owner-Kufarev ODE}
Let us consider the functions
\[
w(z,t)=e^{-t}z\left(1+\sum\limits_{n=1}^{\infty}c_n(t)z^n\right),
\]
satisfying the L\"owner-Kufarev ODE
\begin{equation}\label{LKord1}
\frac{dw}{dt}=-wp(w,t),
\end{equation}
with the initial condition $w(z,0)=z$, and with the function
$p(z,t)=1+p_1(t)z+\dots$ which is holomorphic in $U$ and measurable with respect to
$t\in [0,\infty)$, such that $\re p>0$ in $U$. The function
$w(z,t)$ is univalent and maps $U$ into $U$. 

\begin{lemma}
 Let the function $w(z,t)$ be a solution to the Cauchy problem for the equation
\eqref{LKord1}
with the initial condition $w(z,0)=z$. If the driving function $p(\cdot,t)$, being from the Carath\'eodory class for almost all $t\geq 0$, is  $C^{\infty}$ smooth in the closure $\hat{U}$ of the unit disk $U$ and summable with respect to $t$, then the boundaries of
the domains $B(t)=w(U,t)\subset U$ are smooth for all $t$.
\end{lemma}
\begin{proof}
Observe that the continuous and differentiable dependence of
 the solution to a differential equation $\dot{x}=F(t,x)$ on the initial condition $x(0)=x_0$ is a classical problem. One can refer, e.g., to \cite{Volpato}
in order to assure that summability of $F(\cdot, x)$ regarding to $t$ for each fixed $x$ and  continuous differentiability ($C^1$ with respect to
$x$ for almost all $t$) imply that the solution $x(t,x_0)$ exists, is unique, and is $C^1$ with respect to $x_0$. In our case, the solution to  \eqref{LKord1}
exists, is unique, analytic in $U$, and moreover, $C^1$ on its boundary $S^1$. Let us differentiate \eqref{LKord1} inside the unit disk $U$ with respect to $z$ and write
\[
\log w' =-\int\limits_{0}^{t}(p(w(z,\tau),\tau)+w(z,\tau)p'(w(z,\tau),\tau))d\tau,
\]
choosing the branch of the logarithm such as $\log w'(0,t)=-t$.
This equality is extendable onto $S^1$ because the right-hand side is, and therefore, $w'$ is $C^1$ and $w$ is $C^2$ on $S^1$. We continue analogously and write the formula
\[
w''=-w'\int\limits_{0}^{t}(2w'(z,\tau)p'(w(z,\tau),\tau)
+w(z,\tau)w'(z,\tau)p''(w(z,\tau),\tau))d\tau,
\]
which guarantees that $w$ is $C^3$ on $S^1$. Finally, we come to the conclusion that $w$ is $C^\infty$ on $S^1$.
\end{proof}

Let 
$f(z,t)$ denote $e^tw(z,t)$. The limit $\lim_{t\to\infty}f(z,t)$ is known
 \cite{Pommerenke2} to be a representation of all univalent functions. 
 
  Let the driving term $p(z,t)$ in the L\"owner-Kufarev ODE be from the Carath\'eodory class for almost all $t\geq 0$,   $C^{\infty}$ smooth in $\hat{U}$, and summable with respect to $t$.
 Then the domains $\Omega(t)=w(U,t)$  have  smooth boundary $\partial
 \Omega(t)$. So the L\"owner equation can be extended onto the
 closed unit disk $\hat U=U\cup S^1$.
 
 Consider the Hamiltonian function given by
 \begin{equation}\label{Ham3}
 H=\int\limits_{z\in S^1}f(z,t)(1-p(e^{-t}f(z,t),t))\bar \psi(z,t)\frac{dz}{iz},
 \end{equation}
on the unit circle $z\in S^1$, where $\psi(z,t)$ is a formal series 
$$
\psi(z,t)=\sum_{n=-k}^{\infty}\psi_nz^{n},
$$
defined about the unit circle $S^1$ for any $k\geq 0$.
The Poisson structure on the symplectic space $(f, \bar\psi)$ is given by the canonical brackets
\[
\{P, Q\}=\frac{\delta P}{\delta f}\frac{\delta Q}{\delta \bar \psi}-\frac{\delta P}{\delta \bar \psi}\frac{\delta Q}{\delta f},
\]
or in  coordinate form (only $\psi_n$ for $n\ge 1$ are independent co-vectors corresponding to the tangent vectors $\partial_n$ with respect to the canonical Hermitean product for analytic functions)
\[
\{p, q\}=\sum_{n=1}^{\infty}\frac{\partial p}{\partial c_n}\frac{\partial q}{\partial \bar \psi_n}-\frac{\partial p}{\partial \bar \psi_n}\frac{\partial q}{\partial c_n}.
\]
 Here
\[
P(t)= \int\limits_{z\in S^1}p(z,t)\frac{dz}{iz},\quad Q(t)= \int\limits_{z\in S^1}q(z,t)\frac{dz}{iz}.
\]
The Hamiltonian system becomes
\begin{equation}\label{sys1}
\frac{d f(z,t)}{dt}=f(1-p(e^{-t}f,t))=\frac{\delta H}{\delta \overline{\psi}}=\{f,H\},
\end{equation}
for the position coordinates and
\begin{equation}\label{sys2}
\frac{d\bar \psi}{dt}=-(1-p(e^{-t}f,t)-e^{-t}fp'(e^{-t}f,t))\bar
\psi=\frac{-\delta H}{\delta f}=\{\overline{\psi},H\},
\end{equation}
for the momenta, where $\frac{\delta}{\delta f}$ and $\frac{\delta}{\delta \overline{\psi}}$ are the variational derivatives. So the phase coordinates $(f,\bar{\psi})$ play the role of the canonical Hamiltonian pair.

The coefficients $c_n$ are the complex local coordinates on $\mathcal M$, so in these coordinates we have
\begin{eqnarray}
\dot{c}_n & = &\frac{d c_n}{dt}=c_n-\frac{e^t}{2\pi i}\int\limits_{S^1}w(z,t)p(w(z,t),t)\frac{dz}{z^{n+2}}, \nonumber\\ 
&=&-\frac{1}{2\pi i}\int\limits_{S^1}\sum\limits_{k=1}^ne^{-kt}(e^tw)^{k+1}p_k\frac{dz}{z^{n+2}},\quad n\geq 1.\nonumber
\end{eqnarray}
Let us fix some $n$ and project the infinite dimensional Hamiltonian system on an $n$-dimensional $\mathcal M_n$.
The dynamical equations for momenta governed by the Hamiltonian function \eqref{Ham3} are
\begin{equation*}
\dot{\bar{\psi}}_j=-\bar{\psi}_j+\frac{1}{2\pi i}\sum\limits_{k=1}^n\bar{\psi}_k\int\limits_{S^1}(p+wp')\frac{dz}{z^{k-j+1}},\quad j= 1,\dots, n-1,\label{psi1}
\end{equation*}
and 
\begin{equation}
\dot{\bar{\psi}}_n=0.
\end{equation}
In particular,
\begin{eqnarray*}
\dot{c}_1 & = & -e^{-t}p_1,\\
\dot{c}_2 & = & -2e^{-t}p_1c_1-e^{-2t}p_2,\\
\dot{c}_3 & = & -e^{-t}p_1(2c_2+c_1^2)-3e^{-2t}p_2c_1-e^{-3t}p_3,\\
\dots& & \dots
\end{eqnarray*}
for $n=3$ we have
\begin{eqnarray*}
\dot{\bar{\psi}}_1 & = & 2e^{-t}p_1\bar{\psi}_2+(2e^{-t}p_1c_1+3e^{-2t}p_2)\bar{\psi}_3,\\
\dot{\bar{\psi}}_2 & = & 2e^{-t}p_1\bar{\psi}_3,\\
\dot{\bar{\psi}}_3 & = & 0.
\end{eqnarray*}

 Let us set the function  $L(z):=f'(z,t)\bar\psi(z,t)$. Let $(L(z))_{<0}$ mean the part of the Laurent series for $L(z)$ with negative powers of $z$,
$$
(L(z))_{<0}=(\bar\psi_1+2c_1\bar\psi_2+3c_2\bar\psi_3+\dots)\frac{1}{z}+(\bar\psi_2+2c_1\bar\psi_3+\dots)\frac{1}{z^2}+\dots=\sum\limits_{k=1}^{\infty}\frac{L_k}{z^{k}}.
$$
 Then, the functions $L(z)$ and $(L(z))_{< 0}$ are time-independent
for all $z\in S^1$.

It is easily seen that, passing from the cotangent vectors
$\bar\psi_k$ to the tangent vectors $\partial_k$, the
coefficients $L_k$ of $(L(z))_{<0}$ defined on the tangent bundle  $T\mathcal M^{(1,0)}$ are exactly the Kirillov vector fields $L_k$.  The corresponding fields $L_k$ in the covariant form are
conserved by the L\"owner-Kufarev ODE because $\dot L_k=\{L_k, H\}=0$. The above Poisson structure coincides with  that given by the Witt brackets introduced for $L_k$ previously. For finite-dimensional grades
this result was obtained in \cite{MPV}.

Let us formulate the result as a theorem.

\begin{theorem} Let the driving term $p(z,t)$ in the L\"owner-Kufarev ODE be from the Carath\'eodory class for almost all $t\geq 0$,   $C^{\infty}$ smooth in $\hat{U}$, and summable with respect to $t$. Then the Kirillov  fields in the covariant form are the conserved quantities for the Hamiltonian system
(\ref{sys1}--\ref{sys2}) generated by the L\"owner-Kufarev ODE.
\end{theorem}

\begin{remark}
Another way to construct a Hamiltonian system could be based on the symplectic structure given by the K\"ahlerian form on $\Diff S^1/S^1$. However, there is no explicit expression for such  form in terms of  functions $f\in \Sb$. Moreover, there must be a Hamiltonian formulation
in which the L\"owner-Kufarev equation becomes an evolution equation. This remains an open problem.
\end{remark}

\begin{remark}
At a first glance the situation with an ODE with  a parameter is quite simple. Indeed, if we solve an equation of type $\dot{f}(t, e^{i\theta})=F(f(t, e^{i\theta}),t)$, then fixing~$\theta$ we have an integral  of motion $C=I(f(t,\cdot),t)=\const$. Then, releasing $\theta$, we have $C(e^{i \theta})=I(f(t,e^{i\theta}),t)$. Expanding $C(e^{i\theta})$ into the Fourier series, we obtain an infinite number of conserved quantities, but they do not manifest an
infinite number of degrees of freedom that govern the motion as in the field theory where the governing equations are PDE. In our case, we have not only
one trajectory fixing the initial condition but a pensil of trajectories because our equation has an infinite number of control parameters, the Taylor coefficients of the function $p(z,t)$, which form a bounded non-linear set of admissible controls. Therefore, we operate with sections of the tangent and co-tangent bundles  to the infinite dimensional manifold $\mathcal M$ instead of vector fields along one trajectory as in usual ODE.
\end{remark}

\begin{remark}
No linear combinations $L^*_k$ of $L_1,\dots,L_n,\dots$ allows us to reduce the system of $\{L_k\}$ to a new system of involutory $\{L_k^*\}$ in order
to claim the Liouville integrability of our system. Observe that the coefficients in these linear combinations must be constants to keep conservation laws. 
\end{remark}

\subsection{Construction of $L_0$ and $L_{-n}$}
Consider again the generating function $L(z)=f'(z,t)\bar\psi(z,t)$ and the `non-negative' part  $(L(z))_{\ge 0}$ of the Laurent series for $L(z)$,
 $$
(L(z))_{\ge 0}= (\bar\psi_0+2c_1\bar\psi_1+3c_2\bar\psi_2+\dots)+ (\bar\psi_{-1}+2c_1\bar\psi_0+3c_2\bar\psi_1+\dots)z+\dots
$$
$$
=\sum\limits_{k=0}^{\infty}\mathcal L_{-k}z^{k}.
$$
All $\mathcal L_{-k}$ are conserved by the construction. Define $\bar\psi_0^*=-\sum_{n=1}^{\infty}c_k\bar\psi_k$, and
\[
L_0=\mathcal L_0-(\bar\psi_0-\bar\psi^*_0).
\]
The operator $L_0$ acts on the class $\Sb$ by  $L_0[f](z)=zf'(z)-f(z)$.
 Next define  $L_{-1}=\mathcal L_{-1}-(\bar\psi_{-1}-\bar\psi_{-1}^*)-2c_1(\bar\psi_0-\bar\psi_0^*)$, where $\bar\psi_{-1}^*=0$. Then,
\[
L_{-1}[f](z)=f'(z)-2c_1 f(z)-1
\]
Finally, $$L_{-2}=\mathcal L_{-2}-(\bar\psi_{-2}-\bar\psi_{-2}^*)-2c_1(\bar\psi_{-1}-\bar\psi_{-1}^*)-3{c_2}(\bar\psi_0-\bar\psi_0^*).$$
We choose $\bar\psi_{-2}^*=(c_3-3c_1c_2+c_1^3)\bar\psi_1+\dots$, so that
\[
\bar\psi_{-2}^*[f](z)=\frac{1}{z}-\frac{1}{f(z)}-c_1-(c_2-c_1^2)f(z),
\]
and
\[
L_{-2}[f](z)=\frac{f'(z)}{z}-\frac{1}{f(z)}-3c_1+(c_1^2-4c_2)f(z).
\]
An important fact is that 
\[
L_{0}=c_1 \bar\psi_1+2c_2 \bar\psi_2+\dots,
\]
\[
L_{-1}=(3c_2-2c_1^2)\bar\psi_1+\dots,
\]
\[
L_{-2}=(5c_3-6c_1c2+2c_1^3) \bar\psi_1+\dots,
\]
are linear with respect to $ \bar\psi_k$, $k\geq 1$, and therefore, are sections of $T^*\mathcal M$, which are dual to Kirillov's vector fields. Equivalently, $$L_{0,-1,-2}[f](z)=\mbox{function}(c_1,c_2,\dots)z^2+\dots, \quad z^k=\frac{\partial f}{\partial c_{k-1}}.$$
All other co-vectors we construct by our Poisson brackets as 
\[
L_{-n}=\frac{1}{n-2}\{L_{-n+1},L_{-1}\}=\frac{1}{n-4}\{L_{-n+2},L_{-2}\}.
\]
The form of the Poisson brackets guarantees us that all $L_{-n}$ are linear with respect to $\bar\psi_1,\bar\psi_2,\dots$ and span
the anti-holomorphic part of the co-tangent bundle ${T^{(0,1)}}^*\mathcal M$.

Let us summarize the above in the following conclusion.
We considered a non-linear contour dynamics given by the L\"owner-Kufarev equation. It turned out to be underlined by an algebraic structure, namely, by the Witt algebra spanned by the Virasoro generators $L_n$, $n\in\mathbb Z$.
\begin{itemize}
\item[$\bullet$]   $L_n$, $n = 1,2,\dots$ are the holomorphic Virasoro generators. They span the holomorphic part of the complexiÞed tangent bundle over the space of univalent functions, smooth on the boundary. In the covariant formulation they are conserved by the L\"owner-Kufarev evolution. 

\item[$\bullet$] $L_0$ is the central element.

\item[$\bullet$]  $L_{-n}$, $n = 1,2,\dots$ are the antiholomorphic Virasoro generators. They  span the antiholomorphic part of the decomposition. They 
contain a conserved term and we give an iterative method to obtain them based on 
 the Poisson structure of the L\"owner-Kufarev evolution.
\end{itemize}

\subsection{L\"owner-Kufarev PDE}

 The L\"owner equation in partial derivatives is
$${ \dot{w}(\zeta,t)=\zeta w'(\zeta,t)p(\zeta,t)}, \quad \re p(\zeta,t)>0, \quad |\zeta|<1.$$
with some initial condition $w(z,0)=f_0(z)$.
Let us consider the one-parameter family of functions
$f(z,t)=e^{-t}w(z,t)=z(1+\sum_{n=1}^{\infty}c_n(t)z^n)$, $f(z,0)=f_0(z)$ as a $C^1$ path
in $\Sb$. At the initial point $f_0(z)$ we have that $T_{f_0}\Sb=T_{f_0}\mathcal M^{(1,0)}=T_{f_0}\mathcal M$. A path in the coefficient body $\mathcal M$ in the neighbourhood of $f_0$ is
 $(c_1(t),\dots, c_n(t),\dots)$ with the velocity vector
$\dot c_1\partial_1+\dots+\dot c_n\partial_n+\dots \in T_{f_0}\mathcal M$. 

Taking the Virasoro generators $\{L_k\}$, $k\geq 1$, as a basis in $T_{f_0}\mathcal M^{(1,0)}$ we wish the velocity vector written in this new basis to be
\begin{equation}
{ \dot c_1\partial_1+\dots+\dot c_n\partial_n+\dots=u_1L_1+\dots
u_n L_n+\dots}.\label{e1}
\end{equation}
We compare (\ref{e1})  with the L\"owner-Kufarev equation
\begin{equation} { \dot f=\dot c_1\partial_1+\dots+\dot
c_n\partial_n+\dots=zf'p(z,t)- f=L_0+u_1L_1+\dots u_n L_n+\dots},\label{e2}
\end{equation}
where $p(z,t)=1+u_1z+\dots+u_nz^n+\dots$, and $L_0f=zf'-f$. In view of  similarity between these two expressions \eqref{e1} and \eqref{e2}, we notice that 
\begin{itemize}
\item a new term $L_0$ appears in the L\"owner-Kufarev equation;
\item the function $p(z,t)$ with positive real part corresponds to subordination, whereas for generic trajectories it may have  real part of arbitrary sign. We call this an {\it alternate L\"owner-Kufarev evolution};
\item the vector $L_0$  corresponds exactly to the rotation: $$e^{i\varepsilon}f(e^{-i\varepsilon}z)=f(z)-i\varepsilon(zf'(z)-f(z))+o(\varepsilon).$$  
\end{itemize}

 Let us consider the set $\Sb_0$ of non-normalized smooth univalent functions
of the form $F(z,t)=a_0(t)z+a_1(t)z^2+\dots$, with a tangent vector $\dot a_0\partial_0+\dots+\dot a_n\partial_n+\dots$, where $\partial_k=\partial/\partial a_k$, $k=0,1,2, \dots$. Our aim is to define two different distributions  for the tangent bundle $T\Sb_0$,  that form a sub-bundle of co-dimension 1, which is the tangent bundle
  $T\Sb$. This will be realized by means of formulas
  \eqref{e1} and \eqref{e2}.   Notice that  $\partial_k F=z^{k+1}$. Setting $L_k(F):=z^{k+1}F'$ we get
\[
\dot F=\dot a_0\partial_0+\dots+\dot
a_n\partial_n+\dots=zf'p(z,t)=u_0L_0+u_1L_1+\dots u_n L_n+\dots,
\]
where $p(z,t)=u_0+u_1z+\dots+u_nz^n+\dots$. This alternate L\"owner-Kufarev equation represents
recalculation of the tangent vector in the new basis
\[
{ \dot a_0\partial_0+\dots+\dot a_n\partial_n+\dots=u_0L_0+\dots
u_n L_n+\dots},
\]
where $L_k=a_0\partial_k+2a_1\partial_{k+1}+\dots$.

Let us present the distributions. We start with $F\in \Sb_0$, then we define $f\in \Sb$. The necessary distribution is the map
$$\Sb_0\ni F\to T_f\Sb\hookrightarrow T_F\Sb_0.$$

The analytic form of the first distribution is the following factorization
 $f_1(z,t)=\frac{1}{a_0}F(z,t)=z+\frac{a_1}{a_0}z^2+\dots$, so that
\begin{equation}
 \dot f_1=zf_1'p(z,t)-\frac{\dot a_0}{a_0}f_1,\label{F1}
 \end{equation}
 where $u_0=\frac{\dot a_0}{a_0}$. Then we obtain $$ { \dot c_1\partial_1+\dots+\dot c_n\partial_n+\dots=\hat L_0+u_1\hat L_1+ \dots +u_n \hat L_n+\dots}$$
where $\hat L_0f_1=u_0(zf_1'-f_1)$, $\hat L_kf_1=z^{k+1}f_1'$, $c_k=\frac{a_k}{a_0}$,
$\partial_k=\frac{\partial}{\partial c_k}$.  In
particular, $a_0=e^t$ implies the L\"owner-Kufarev equation for arbitrary sign of $\re p$.

The analytic form of the second distribution becomes
 $f_2(z,t)=F(\frac{1}{a_0}z,t)=z+\frac{a_1}{a^2_0}z^2+\dots$, so that
\begin{equation}
 \dot f_2=zf_2'p(\frac{z}{a_0},t)- \frac{\dot a_0}{a_0} zf'_2,\label{F2}
 \end{equation}
 where again $u_0=\frac{\dot a_0}{a_0}$.
In the coefficient form we get
$$ { \dot c_1\partial_1+\dots+\dot c_n\partial_n+\dots=u_1\tilde L_1+ \dots +u_n \tilde L_n+\dots}$$
where  $\tilde L_kf_2=z^{k+1}f_2'$, $c_k=\frac{a_k}{a^{k+1}_0}$,
$\partial_k=\frac{\partial}{\partial c_k}$.

Observe that the equation  \eqref{F2} gives an identification of $T\mathcal M^{(1,0)}$ with  $T\mathcal M$.

Finally, let us make an explicit calculation of $\hat{L}_0$, which for $a_0=e^t$ we continue to denote by $L_0$. Using Kirilov's basis $L_1,L_2,\dots$ as a linear combination we write
$$L_0=\sum_{m=1}^{\infty}\Pi_mL_m.$$ The coefficients $\Pi_m$ are polynomials, which can be
obtained  using the following recurrent formulas
$$K_1=0,\quad K_m=-\sum_{j=1}^{m-1}j(m-j+1)c_{m-j}c_j,\quad
\Pi_m=mc_m+\sum_{j=1}^{m}K_{m-j+1}P_{j-1},$$ where $P_k$ are
polynomials
\begin{equation}\label{pol1}
P_0=1,\quad P_1=-2c_1,\quad P_2=4c^2_1-3c_2,\quad P_k=-\sum_{j=1}^{k}(j+1)c_jP_{k-j},
\end{equation} 

Let us summarize the above considerations in the following theorem.

\begin{theorem}
The L\"owner-Kufarev PDE \eqref{F1} gives the distribution for the tangent bundle $T\Sb_0$ of non-normalized smooth univalent functions  $\Sb_0$,  that forms a sub-bundle of co-dimension 1, which is  the tangent bundle
  $T\Sb$. 
  
The equation \eqref{F2} gives another distribution, and moreover,  
 it makes the explicit correspondence between the natural  complex structure 
of   $T\Sb$, as $\Sb$ embedded into $\mathbb C^{\mathbb N}$, and the complex structure of $T\mathcal M^{(1,0)}$ at each point $f\in \Sb$ defined by \eqref{CompStruct}.
\end{theorem}

One of the reason to consider the alternate L\"owner-Kufarev PDE is the regularized canonical Brownian motion on smooth Jordan curves.
For all Sobolev metrics $H^{\frac{3}{2}+\varepsilon}$, the classical theory of stochastic flows allows to construct Brownian motions on $C^1$ diffeomorphism group of $S^1$. The case 3/ 2 is critical. Malliavin \cite{Malliavin} constructed the canonical Brownian motion on the Lie algebra $\Vect S^1$ for the Sobolev norm $H^{3/2}$. Another construction was proposed in \cite{Fang}. Airault and Ren \cite{AiraultRen} proved that the infinitesimal version
of the Brownian flow is H\"older continuous with any exponent $\beta<1$.

The regularized canonical Brownian motion on $\Diff S^1$ is a stochastic flow on $S^1$ associated to the It{\^{o}} stochastic differential equation
\[
dg^r_{x,t}=d\zeta^r_{x,t}(g^r_{x,t}),
\]
\[
\zeta^r_{x,t}(\theta)=\sum_{n=1}^{\infty}\frac{r^{n}}{\sqrt{n^3-n}}(x_{2n}(t)\cos n\theta-x_{2n-1}(t)\sin n\theta),
\]
where $\{x_k\}$ is a sequence of independent real-valued Brownian motions and $r\in (0,1)$ and the series for $\zeta^r_{x,t}(\theta)$ is a Gaussian trigonometric series. Kunita's theory of stochastic flows asserts that
the mapping $\theta\to g^r_{x,t}(\theta)$ is a $C^{\infty}$ diffeomorphism and the limit $\lim\limits_{r\to 1^{-}}g^r_{x,t}=g_{x,t}$ exists uniformly
in $\theta$. The random homeomorphism $g_{x,t}$ is called {\it canonical Brownian motion} on $\Diff S^1$, see \cite{AiraultRen, Fang, Malliavin, RenZhang}. It was shown in \cite{AiraultRen, Fang}, that this random homeomorphism is H\"older continuous.

The canonical Brownian motion can be defined not only on $\Diff S^1$, but  also on the space of $C^{\infty}$-smooth Jordan curves by conformal welding. This leads to dynamics of random loops which are not subordinated.

\section{Elliptic operators over the coefficient body}

The Kirillov first order differential operators $L_k$ generate the elliptic operator $\sum |L_k|^2$. In this section we construct the geodesic equation and find geodesics with constant velocity coordinates in the field of this operator. In particular, we shall prove that the norm of the driving function in the L\"owner-Kufarev theory gives the minimal energy of the motion
in this field.

\subsection{Dynamics within the coefficient body}
Let us  recall the geometry of the coefficient body  $\mathcal M_n$ for finite $n$. The affine coordinates are introduced by projecting
$$
\mathcal M\ni f=z\Big(1+\sum\limits_{k=1}^{\infty}c_kz^k\Big)\mapsto (c_1,\ldots, c_n)\in\mathcal M_n.
$$
The manifold $\mathcal M_n$ was studied actively in the middle of the last century, see e.g., \cite{Babenko, Schaeffer}.
We compile some important properties of $\mathcal M_n$ below:
\begin{itemize}

\item[(i)] $\mathcal M_n$ is homeomorphic to a $(2n-2)$-dimensional ball and its
boundary $\partial \mathcal M_n$ is homeomorphic to a $(2n-3)$-dimensional
sphere;

\item[(ii)] every point $x\in \partial \mathcal M_n$ corresponds to exactly one function $f\in
\Sb$ which is called a {\it boundary function} for $\mathcal M_n$;

\item[(iii)] boundary functions map the unit disk $U$ onto the complex plane
$\mathbb C$ minus piecewise analytic Jordan arcs forming a tree with
a root at infinity and having at most $n$ tips,

\item[(iv)] with the exception for a set of smaller dimension,
at every point $x\in \partial \mathcal M_n$ there exists a normal vector
satisfying the Lipschitz condition;

\item[(v)] there exists a connected open set $X_1$ on $\partial \mathcal M_n$,
such that the boundary $\partial \mathcal M_n$ is an analytic hypersurface at
every point of $X_1$. The points of $\partial \mathcal M_n$ corresponding to
the functions that give the extremum to a linear functional belong
to the closure of $X_1$.

\end{itemize}

Properties (ii) and (iii) imply that the functions from $\Sb$ deliver interior points
of $\mathcal M_n$. The Kirillov operators $L_j$ restricted onto $\mathcal M_n$
give truncated vector fields
$$
L_j=\partial_j+\sum\limits_{k=1}^{n-j}(k+1)c_k\partial_{j+k},
$$
which we, if it causes no confusion, continue denoting by $L_j$ in this section.
In~\cite{MPV} based on the L\"owner-Kufarev representation, we showed that these $L_j$ can be obtained from a partially integrable Hamiltonian system for the coefficients in which the first integrals coincide with $L_j$.

Let $c(t)=\big(c_1(t),\ldots,c_{n}(t)\big)$ be a smooth trajectory in $\mathcal M_n$; that is a $C^1$ map $c:[0,1]\to \mathcal M_n$. Then the velocity vector $\dot c(t)$ written in the affine basis as $\dot c(t)=\dot c_1(t)\p_1+\ldots+\dot c_{n}(t)\p_{n}$ can be also represented in the basis of vector fields $L_1,\ldots,L_{n}$ (compare with \eqref{F2}) as \begin{eqnarray}\label{eq4}
\dot c(t) & =\dot c_1(t)\p_1+\ldots+\dot c_{n}(t)\p_{n}\\                                             & = u_1L_1+u_2L_3+\ldots+u_{n}L_{n}, \nonumber                                                                       \end{eqnarray}
where the coefficients $u_k$ can be written in the recurrent form  as 
\begin{equation}\label{eq3}
u_1=\dot c_1,\qquad u_k=\dot c_k-\sum_{j=1}^{k-1}(j+1)\dot c_ju_{k-j}.
\end{equation} Expressing  $u_k$ in terms of $c_k$ and $\dot c_k$, we get
\begin{equation}\label{eq8}
u_k=\dot c_k+\sum_{j=1}^{k-1}P_{j}\dot c_{k-j}.
\end{equation}

One may notice that these polynomials are the first coefficients of the holomorphic function $1/f'(z)$, where $f\in \Sb$. In the infinite dimensional case this follows from the L\"owner-Kufarev equation
 \eqref{F2} with $a_0=e^t$. Kirillov's fields $L_k$ act over these polynomials as  
 $$
 L_kP_n=(n-2k-1)P_{n-k}\quad n\geq k\quad\text{and}\quad L_kP_n=0 \quad n< k.
 $$

\begin{proposition}
 We define
\begin{eqnarray}\label{eq2}
\omega_1 & = & dc_1,\nonumber \\
\omega_2 & = & dc_2-2c_1\omega_1,\nonumber\\
\ldots & \ldots & \ldots\ldots\ldots\ldots\ldots,\nonumber\\
\omega_n & = & dc_n-\sum_{j=1}^{n-1}(j+1)c_j\omega_{n-j}.
\end{eqnarray}
Then, $\{\omega_1,\dots,\omega_n\}$ is a conjugate to $\{L_1,\ldots,L_n\}$ basis of one-forms. Namely,
$$
\omega_n(L_n)=1,\quad \omega_n(L_k)=0\ \ \text{if}\ \ k\neq n.
$$ 
\end{proposition}
\begin{proof}
If $k>n$, then the vector fields $L_k$ do not contain $\p_n$. Since the form $\omega_n$ depends only on $dc_j$ with $j<n$, then $$\omega_n(L_k)=\p_n(L_k)-\sum_{j=1}^{n-1}(j+1)c_j\omega_{n-j}(L_k)=0\ \ \text{for}\ \ k>n>n-j.$$ If $n=k$, then $$\omega_n(L_n)=\p_n(L_n)-\sum_{j=1}^{n-1}(j+1)c_j\omega_{n-j}(L_n)=1+0\ \ \text{for}\ \ n>n-j.$$ To prove the case $k<n$ we apply the induction. Let us show for $L_1$. We have $$\omega_2(L_1)=dc_2(L_1)-2c_1(L_1)=2c_1-2c_1=0.$$ We suppose that $\omega_n(L_1)=0$. Then $$\omega_{n+1}(L_1)=dc_{n+1}(L_1)-\sum_{j=1}^{n}(j+1)c_j\omega_{n+1-j}(L_1)=(n+1)c_{n}-(n+1)c_n\omega_1(L_1)=0.$$ The same arguments work for $\omega_n(L_k)$ with $k<n$.
\end{proof}

In the affine basis the forms $\omega_k$  can be written making use of the
polynomials $P_n$. We observe that one-forms $\omega_k$ are defined
in a similar way as the coordinates $u_k$ with respect to
the Kirillov vector fields $L_k$. Thus, if we develop the recurrent
relations~\eqref{eq2} and collect the terms with $dc_n$ we get
$$\omega_k=dc_k+\sum\limits_{j=1}^{k-1}P_jdc_{k-j}.\quad
k=1,\ldots,n.$$

By the duality of tangent and co-tangent bundles the information about the motion is encoded by these one-forms.

\subsection{Hamiltonian equations}
There exists an Hermitian form on $T\mathcal M_n$, such that the system
$\{L_1,\ldots,L_n\}$ is orthonormal with respect to this form. The
operator $L=\sum |L_k|^2$ is  elliptic, and  we write the
Hamiltonian function $H(c,\bar c, \psi, \bar\psi)$ defined on the
co-tangent bundle, corresponding to the operator $L$ as $H(c,\bar c,
\psi, \bar\psi)=\sum_{k=1}^{n}|l_k|^2$, where
$$l_k=\bar \psi_k+\sum_{j=1}^{n-k}(j+1)c_j\bar \psi_{k+j}.$$ The
corresponding Hamiltonian system admits the form \begin{eqnarray*}
\dot c_1 & = & \frac{\p H}{\p\bar\psi_1}= \bar l_1 \\
\ldots & = & \ldots\ldots\ldots\ldots \\
\dot c_n & = & \frac{\p H}{\p\bar\psi_n}=\bar l_n+\sum_{j=1}^{n-1}(j+1) c_j\bar l_{n-j}\\
\dot{\bar\psi}_p & = & -\frac{\p H}{\p c_p}=-(p+1)\sum_{k=1}^{n-p} l_k\bar \psi_{k+p}\\
\ldots & = & \ldots\ldots\ldots\ldots \\
\dot{\bar\psi}_n & = & -\frac{\p H}{\p c_n}=0.
\end{eqnarray*} Let us observe that
\begin{equation}\label{eq6}\dot l_k=\sum_{j=1}^{n-k}(j-k)\bar l_jl_{j+k}.\end{equation}
Expressing $\bar l_k$ from the first $n$ Hamiltonian equations we get
\begin{equation}\label{eq7}\bar l_k=\dot c_k+\sum_{j=1}^{k-1}P_{j}\dot c_{k-j},\quad k=1,\ldots,n.\end{equation} We can decouple the Hamiltonian system making use of~\eqref{eq6} and~\eqref{eq7} which leads us to the following non-linear differential equations of the second order
$$\ddot c_k=\dot{\bar l}_k+\sum_{j=1}^{k-1}(j+1)c_j\dot{\bar  l}_{k-l}+\sum_{j=1}^{k-1}(j+1)
\dot c_j\bar l_{k-l},$$ where $\dot{ l}_k$ are expressed in terms of the
product of $\bar l_jl_{j+k}$ by~\eqref{eq6}, and the last products depend on
$P_j$, $\bar P_j$ and $\dot c$, $\dot{\bar  c}_j$ for the
corresponding indices $j$ by \eqref{eq7}. For example,
$$
\ddot c_1  = \dot{\bar  l}_1=\sum_{j=1}^{n-1}(j-1)\Big(\dot
c_j+\sum_{p=1}^{j-1}P_p \dot c_{j-p}\Big)\overline{\Big(\dot
c_{j+1}+\sum_{q=1}^{j}P_q\dot c_{j+1-q}\Big)}.
$$

Comparing~\eqref{eq7} and~\eqref{eq8}, we conclude that $\bar
l_k=u_k$ and $u_k$ satisfy the differential equations
\begin{equation}\label{eq9}\dot u_k=\sum_{j=1}^{n-k}(j-k)\bar u_ju_{j+k},\end{equation} 
on the solution of the Hamiltonian system.
Observe that any solution of~\eqref{eq9} has a velocity vector of constant length. It is easy to see from the following system
\begin{eqnarray}\label{skewsym}
\bar u_1\dot u_1 & = & 0 \bar u_1\bar u_1u_2+\bar u_1\bar u_2u_3+2\bar u_1\bar u_3u_4+3\bar u_1\bar u_4u_5+4\bar u_1\bar u_5u_6+\ldots,\nonumber\\
\bar u_2\dot u_2 & = & -1 \bar u_1\bar u_2u_3+ 0 \bar u_2\bar u_2u_4+1\bar u_2\bar u_3u_5+2\bar u_2\bar u_4u_6+\ldots,\nonumber\\
\bar u_3\dot u_3 & = & -2 \bar u_1\bar u_3u_4-1 \bar u_2\bar u_3u_5+0\bar u_3\bar u_3u_6+\ldots,\\
\bar u_4\dot u_4 & = & -3 \bar u_1\bar u_4u_5-2 \bar u_2\bar u_4u_6+\ldots,\nonumber\\
\bar u_5\dot u_5 & = & -4 \bar u_1\bar u_5u_6+\ldots,\nonumber\\
\bar u_6\dot u_6 & = & \ldots\nonumber
\end{eqnarray} Then, $$\frac{d|u|^2}{dt}=2\sum_{k=1}^{n}(\bar u_k\dot u_k+u_k\dot{\bar u}_k)=0,$$ for any $n$, thanks to the cut form of our vector fields and the skew symmetry of \eqref{skewsym}. The simplest solution may be deduced for constant driving terms  $u_k$, $k=1,\ldots,n$. The Hamiltonian system immediately gives the geodesic
\begin{eqnarray*}
c_1 & = & \bar u_1(0)s+c_1(0),\\
c_2 & = & \bar u_1^2(0)s^2+\bar u_2(0)s+c_2(0),\\
c_3 & = & 3\bar u_1(0)\big(\bar u_1^2(0)\frac{s^3}{3}+\bar u_2(0)\frac{s^2}{2}+c_2(0)\big)+2\bar u_2(0)
\big(\bar u_1(0)\frac{s^2}{2}+c_1(0)s\big)+\bar u_3(0)s+c_2(0),\\
\ldots & = & \ldots\ldots\ldots\ldots\ldots\ldots
\end{eqnarray*}
In general, $c_n$ becomes a polynomial of order $n$ with coefficients
that depend on the initial data $c(0)$ and on the initial velocities
$\bar u(0)$.

The Lagrangian $\mathcal L$ corresponding to the Hamiltonian function $H$  can
be defined by the Legendre transform as $$\mathcal L=(\dot c,\bar
\psi)-H=\sum_{k=1}^{n}\Big(\bar l_k\bar \psi_k+\bar
\psi_k\sum_{j=1}^{k-1}(j+1)c_j\bar
l_{k-j}\Big)-\frac{1}{2}\sum_{k=1}^{n}|l_k|^2.$$ Taking into account that 
\[
\bar \psi_k\dot c_k =\sum\limits_{j=1}^{k-1} (j+1)c_j\bar{\psi}_k\bar{l}_{k-j}+\bar{\psi}_k\bar{l}_k.
\]
Summing up over $k$, we obtain $(\dot c,\bar
\psi)=\sum_{k=1}^{n}l_k\bar l_k=\sum_{k=1}^{n}\bar u_ku_k$, that
gives us
$$\mathcal L(c,\dot c)=\frac{1}{2}\sum_{k=1}^{n}|u_k|^2.$$ 
All these considerations can be generalized for $n\to\infty$.
Thus, we conclude that the coefficients of the function $p(z,t)$ in the L\"owner-Kufarev PDE play the
role of generalized moments for the dynamics in $\mathcal M_n$ and $\mathcal M$ with respect to the Kirillov basis on the tangent bundle. Moreover,
the $L^2$-norm of the function $p$ on the circle $S^1$ is the energy of such motion.

\section{SLE and CFT}

In this section we briefly review for completeness the connections between conformal field theory (CFT)  and  Schramm-L\"owner evolution (SLE) following, e.g., \cite{BB, FriedrichWerner}). SLE (being, e.g., a continuous limit of CFT's archetypical Ising model at its critical point) gives an approach to CFT which emphasizes CFT's roots in statistical physics. 

SLE$_{\varkappa}$ is a $\varkappa$-parameter family of covariant processes describing the evolution of random sets called the SLE$_{\varkappa}$ hulls. For different values of $\varkappa$ these sets can be either a simple fractal curve $\varkappa\in [0,4]$, or a self-touching curve $\varkappa\in (4,8)$, or a space filling Peano curve $\varkappa\geq 8$. At this step we deal with the chordal version of SLE. The complement to a SLE$_{\varkappa}$ hull in the upper half-plane $\mathbb H$ is a simply connected domain that is mapped conformally onto $\mathbb H$ by a holomorphic function $g(z,t)$ satisfying the equation
\begin{equation}\label{SLE}
\frac{dg}{dt}=\frac{2}{g(z,t)-\xi_t}, \quad g(z,0)=z,
\end{equation}
where $\xi_t=\sqrt{\varkappa}B_t$, and $B_t$ is a normalized Brownian motion with the diffusion constant $\varkappa$. The function $g(z,t)$ is expanded as $\displaystyle g(z,t)=z+\frac{2t}{z}+\dots$. The equation \eqref{SLE} is called the Schramm-L\"owner equation and was studied first in \cite{LSW1}--\cite{LSW3}, see also \cite{RohdeSchramm01} for basic properties of SLE.
Special values of $\varkappa$ correspond to interesting special cases of SLE, for example $\varkappa=2$ corresponds to the loop-erasing random walk and the uniform spanning tree, $\varkappa=4$ corresponds to the harmonic explorer and the Gaussian free field. Observe, that the equation \eqref{SLE} is not a stochastic differential equation (SDE). To rewrite it in a stochastic way (following \cite{BB}, \cite{FriedrichWerner}) let us set a function $k_t(z)=g(z,t)-\xi_t$, where $k_t(z)$ satisfies already the SDE
\[
dk_t(z)=\frac{2}{k_t(z)}dt-d\xi_t.
\]
For a function $F(z)$ defined in the upper half-plane one can derive the It\^o differential
\begin{equation}\label{Ito}
dF(k_t)=-d\xi_t L_{-1}F(k_t)+dt (\frac{\varkappa}{2}L_{-1}^2-2L_{-2})F(k_t), 
\end{equation}
with the operators $L_{-1}=-\frac{d}{d z}$ and $L_{-2}=-\frac{1}{z}\frac{d}{d z}$. These operators  are the first two Virasoro generators in the `negative' part of the Witt algebra spanned by the operators $-z^{n+1}\frac{d}{d z}$ acting on the appropriate representation space. All other generators can be obtained by the commutation relation
\[
[L_m,L_n]=(n-m)L_{n+m}.
\]For any state $|\psi\rangle$, the state $L_{-1}|\psi\rangle$ measures the diffusion of $|\psi\rangle$ under SLE, and
$(\frac{\varkappa}{2}L_{-1}^2-2L_{-2})|\psi\rangle$ measures the drift. The states of interest are drift-less, i.e., the second term in \eqref{Ito} vanishes. Such states are annihilated by $\frac{\varkappa}{2}L_{-1}^2-2L_{-2}$, which is true if we choose the state $|\psi\rangle$ as the highest weight vector  in the highest weight representation of the Virasoro algebra  with the central charge $c$ and the conformal weight $h$ given by
\[
c=\frac{(6-\varkappa)(3\kappa-8)}{2\varkappa},\quad h=\frac{6-\varkappa}{2\varkappa},
\]
and the operators $L_{-1}$ and $L_{-2}$ are  taken in the corresponding representation. It was obtained in  \cite{BB} and \cite{FriedrichWerner}, that $F(k_t)$ is a martingale if and only if
 $(\frac{\varkappa}{2}L_{-1}^2-2L_{-2})F(k_t)=0$.
We define a CFT with a boundary in $\mathbb H$ such that the boundary condition is changed by a boundary operator. The random curve in $\mathbb H$ defined by SLE is growing so that
it has states of one type to the left and of the other type to the right (the simplest way to view this is the lattice Ising model with the states defined as spin positions up or down). The mapping $g$ satisfying \eqref{SLE} `unzips' the boundary. The primary operator that induces the boundary change
with the conformal weight $h$ is drift-less, and therefore, its expectation value does not change in time under the boundary unzipping. Hence all correlators computing with this operator remain invariant. Analogous considerations one may provide for the `radial' version of SLE in the unit disk, slightly modifying the above statements.

\end{document}